\documentclass[journal]{IEEEtran}

\usepackage[T1]{fontenc}
\usepackage{cite}
\usepackage{amsmath,amssymb,amsfonts}
\usepackage{amsthm}
\usepackage{graphicx}
\usepackage{tikz}
\usetikzlibrary{arrows.meta,positioning}
\usepackage{dcolumn}
\usepackage{bm}
\usepackage{xcolor}
\usepackage{soul}
\usepackage{comment}
\usepackage[english]{babel}
\usepackage[linesnumbered,ruled,vlined]{algorithm2e}
\usepackage{algpseudocode}
\usepackage{acronym}
\usepackage{cleveref}
\acrodef{bdd}[BDD]{bad data detector}
\acrodef{bic}[BIC]{Bayesian information criterion}
\acrodef{cm}[CM]{cycle manifold}
\acrodef{cs}[CS]{cycle space}
\acrodef{csd}[CSD]{cycle space detector}
\acrodef{fdia}[FDIA]{false data injection attack}
\acrodef{gpu}[GPU]{graphics processing unit}
\acrodef{mas}[MAS]{Minimal Attack Set}
\acrodef{mdbs}[MDBS]{Minimum Defender Blocking Set}
\acrodef{pca}[PCA]{principal component analysis}
\acrodef{rms}[RMS]{root mean square}
\acrodef{rmt}[RMT]{random matrix theory}
\acrodef{rrqr}[RRQR]{rank-revealing QR factorization}
\acrodef{svd}[SVD]{singular value decomposition}
\acrodef{tls}[TLS]{total least squares}
\newtheorem{theorem}{Theorem}

\newtheorem{corollary}[theorem]{Corollary}
\theoremstyle{definition}
\newtheorem{definition}{Definition}
\theoremstyle{remark}

\DeclareMathOperator{\rank}{rank}

\def\graph{\mathbb{G}}
\def\E{E}
\def\V{V}
\def\v{v}
\def\e{e}
\def\noiseE{\mathbf{E}}

\def\dimx{n}
\def\dimz{m}
\def\time{T}
\def\busi{i}
\def\busj{j}
\def\P{\mathbf{P}}

\def\reac{x}
\def\angle{\theta}
\def\G{\mathbf{G}}
\def\h{h}
\def\A{\mathbf{A}}
\def\row{\mathbb{R}_{\rho}}
\def\nul{\mathbb{N}}

\def\z{\vec{z}}
\def\Z{\mathbf{Z}}
\def\S{\mathbf{S}}

\def\R{\mathbb{R}}
\def\I{\mathbf{I}}
\def\const{k}
\def\cyc{\mathbb{C}}
\def\c{c} 

\def\bi{\mathbb{B}}

\def\C{C}

\def\thresh{\tau} 

\def\H{\mathbf{H}}

\def\noise{e}
\def\p{p}

\def\x{\vec{x}}
\def\X{\mathbf{X}}
\def\n{\vec{n}}

\def\ex{\hat{\mathbf{x}}}

\def\eR{\mathbf{R}}

\def\jc#1{{}}

\newcommand{\rev}[1]{{\color{blue}#1}}

\begin{document}

\title{From Cycle Space to Cycle Manifold:\\ Limits and Achievability of Blind False Data Injection Attacks}

\IEEEaftertitletext{}
\author{\IEEEauthorblockN{
\textbf{Xin Li}\IEEEauthorrefmark{1},
\textbf{Chenhan Xiao}\IEEEauthorrefmark{2}, 
\textbf{Jonathan Cohen}\IEEEauthorrefmark{1}, 
\textbf{Aviad Elyashar}\IEEEauthorrefmark{3}\IEEEauthorrefmark{4}, 
\textbf{Yang Weng}\IEEEauthorrefmark{2},
\textbf{Rami Puzis}\IEEEauthorrefmark{1}\IEEEauthorrefmark{4}
}

\IEEEauthorblockA{\IEEEauthorrefmark{1}}Faculty of Computer and Information Science, Ben-Gurion-University, Be'er Sheva, Israel \\
\IEEEauthorblockA{\IEEEauthorrefmark{2}School of Electrical, Computer and Energy Engineering, Arizona State University, AZ, USA}\\
\IEEEauthorblockA{\IEEEauthorrefmark{3}Department of Computer Science, Shamoon College of Engineering, Be'er Sheva, Israel}\\
\IEEEauthorblockA{\IEEEauthorrefmark{4}Cyber@BGU, Ben-Gurion University of Negev, Be'er Sheva, Israel}\\
}

\maketitle

\begin{abstract}

A false data injection attack (FDIA) can change the estimated grid state while
evading a residual-based bad-data detector (BDD).  Existing blind attacks
learn a low-rank measurement subspace, but this algebraic view does not state
the physical grid constraints that make an attack stealthy or the minimum
information needed to recover the complete attack space.  Under the connected
direct-current (DC) branch-flow model, we show that the residual-sensitive subspace of the
noiseless orthogonal test is exactly the weighted cycle space.  Its orthogonal
complement is therefore the complete stealthy attack space, making weighted
cycle-space knowledge both necessary and sufficient for complete blind FDIA.
This space identifies the
topology only up to 2-isomorphism and the relative cycle-edge parameters only
up to one  scale per biconnected component; bridge parameters are neither
identified nor required.  We then formulate a computationally unconstrained
benchmark and a tractable measurement-only reconstruction method.  Experiments
on IEEE systems compare BDD bypass rate at a 95\% nominal-acceptance threshold
against state impact.  As a compact alternating-current (AC) extension, we characterize feasible
 branch $P/Q$ measurements by a cycle manifold and demonstrate
topology-assisted manifold fitting and measurement generation on a graphics
processing unit (GPU).  In
the lossless fixed-voltage small-angle limit, the normal space of the
active-power slice reduces to the DC weighted cycle space.

\end{abstract}



\acresetall
\section{Introduction}
Power system state estimation turns  measurements into the bus states used for grid monitoring and control, while a residual-based
\ac{bdd} rejects measurements that are inconsistent with the estimated state \cite{abur2004power}. 
On the attacker's side, given sufficient system information about the power grid, an \ac{fdia} can bias the state without causing a corresponding increase in residuals, so the control center may accept an incorrect view of the grid \cite{liu2009false,hug2012vulnerability}.
Luckily, the power grid's system information is well protected.
So the attacker cannot easily inject a false measurement without an \ac{bdd} alarm.
The most feasible information for the attacker is the time series measurement.
The measurement contains the normal operation of the power grid.
Extracting useful information to craft a stealthy injection from normal measurement is called blind \ac{fdia}.
blind \ac{fdia} is the most realistic threat to the modern power grid because it requires minimal information.

Blind \ac{fdia} replaces knowledge of the system with historical measurements.
Existing methods estimate the attack space through subspace learning
\cite{kim2014subspace}, PCA \cite{yu2015blind}, random-matrix perturbation
\cite{lakshminarayana2020data}, or matrix reconstruction
\cite{yang2022blind}. 
These methods show that a low-rank measurement structure can support blind \ac{fdia}, but the learned object typically captures only a small subset of the attack directions rather than the complete attack space.

Besides these subspace methods, some topology and physics-related work is also proposed.
Martin et al.~\cite{higgins2021topology} try to estimate the full system information using the method in \cite{zhang2020topology}.
The information required to estimate the full system is not only the branch power flow but also the state node voltage.
Chin et al.~\cite{chin2017blind} propose a linear single attack direction estimation method with the assumption that the difference of state deviation is small.
Some deep learning and machine learning methods have also been proposed for blind \ac{fdia}. But these methods have a minor influence.

Since the early development of \ac{fdia}~\cite{abur2004power}, a large body of research has proposed increasingly sophisticated approaches for constructing \ac{fdia}.
Yet several fundamental questions remain insufficiently understood: What is the physical origin of \ac{fdia} stealthiness? What underlying structure gives residual-based \ac{bdd} its protection capability? And what fundamentally determines the limits and achievability of  \ac{fdia}?
In this paper, we address these questions from a unified structural perspective and develop a deeper understanding of both \ac{bdd} and \ac{fdia}.
We show that the key lies in the network's cycle structure: in the DC model, this structure appears as the linear cycle space, while in the AC model, it generalizes to a nonlinear cycle manifold.
From this viewpoint, residual protection, stealthy perturbations, and the fundamental limits of blind attacks can all be understood through the same underlying principle.

Cycle physics is established separately in the power system and graph literature.
Cycle variables express power-flow consistency \cite{horsch2018linear,farivar2013branch}, while the graph cycle space determines topology only up to 2-isomorphism
\cite{whitney19332,gross2018graph}.
These works do not establish the relationship between \ac{fdia} and cycle structures.

In this paper, we establish the explicit physical constraints for both DC and AC power models. 
We reveal that the cycle structure is the fundamental component for both AC and AC models.
In the DC model, the \ac{fdia} is constrained by the weighted cycle space.
In the AC model, the \ac{fdia} is constrained by the cycle manifold.
We prove that the weighted cycle space and cycle manifold are sufficient and necessary conditions to perform a complete \ac{fdia} for DC and AC models, respectively.
We therefore study the limits and achievability of blind \ac{fdia} using these physical constraints.
With these theories and studies, we not only provide a new tool for \ac{fdia} research but also offer new insights into power grid planning and protection.

The main contributions are as follows:
\begin{itemize}
    \item \textbf{Physical origin of \ac{fdia} stealth.}
    We prove that the weighted cycle space is the exact \footnote{Exact means necessary and sufficient} orthogonal complement of the complete attack space in the DC model.
    For the AC model, we propose the theory of \acl{cm} and prove that it is the exact nonlinear constraint for the complete attack space. 
    With linearization and a small voltage angle assumption, the \acl{cm} can be transformed to the constraint in weighted \acl{cs}.

    \item \textbf{Limits and achievability}
    Estimating the weighted \acl{cs} and \acl{cm} both need a cycle basis and a parameter estimation process.
    We develop a tractable measurement-only DC weighted \acl{cs} reconstruction method.
    With topology and measurement information, we can obtain an upper bound for \acl{cs} estimation.
    For the AC model, we show that it is almost impossible to estimate a correct cycle basis because it is a nonlinear combinatorial problem.
    With topology and measurement information, we develop a zero-bias \acl{cm} estimation method. 
    Both methods outperform the existing blind \ac{fdia} methods.

    \item \textbf{Planning and protection implications.}
    Given the limits and achievability, we provide several guidelines for power grid planning and protection.
    
\end{itemize}

Our earlier work uses cycle-space reconstruction for detection and studies its reconstruction and generalization errors \cite{li2026cycle}.
Here, we use the cycle view to establish the physical origin of the \ac{fdia}, the minimum attacker information, the computational limit, and DC/AC attack realizations.


\section{Preliminaries}
\label{sec:preliminaries}

\subsection{State estimation, \ac{bdd} and blind \ac{fdia}}

We represent the operational transmission grid by a graph $\G=\graph(\E,\V)$ with $\dimz=|\E|$ branches and $\dimx=|\V|$ buses. $q=\dimz-\dimx+1$ is the number of cycles in the grid.

Under the AC model, the reference-bus voltage magnitude and angle are fixed.
Let $\x\in\R^{2\dimx}$ contain the voltage magnitudes and angles.
The measurements satisfy
\begin{equation}
\label{eq:ac_measurement_model}
\z=\h(\x)+\noise,
\end{equation}
where $\h(\cdot)$ is the nonlinear AC measurement function and $\noise$ is
measurement noise.  The weighted least-squares state estimate is
\begin{equation}
\label{eq:ac_state_estimation}
\ex=\arg\min_{\x}
\bigl[\z-\h(\x)\bigr]^{T}\eR^{-1}
\bigl[\z-\h(\x)\bigr],
\end{equation}
where $\eR$ is the measurement-noise covariance. 
The \ac{bdd} uses
\begin{equation}
\label{eq:ac_bdd_statistic}
J_{\mathrm{AC}}(\z)
=\bigl[\z-\h(\ex)\bigr]^{T}\eR^{-1}
\bigl[\z-\h(\ex)\bigr].
\end{equation}
It rejects the measurements when $J_{\mathrm{AC}}(\z)>\thresh$ and accepts them otherwise.

Under the DC approximation, the active branch-flow measurement on line $(\busi,\busj)$ is
\begin{equation}
\label{eq: dc_power_measurement}
\P_{\busi\busj}=\frac{\angle_{\busi\busj}}{\reac_{\busi\busj}},
\end{equation}
where $\angle_{\busi\busj}=\angle_\busi-\angle_\busj$.
Stacking all branch-flow measurements gives
\begin{equation}
\label{eq: dc_power_flow}
\z=\H\angle+\noise
\end{equation}
where $\angle\in\R^{\dimx}$ is the voltage-angle state and $\noise\in\R^{\dimz}$ is measurement noise.
The Jacobian decomposes as
\begin{equation}
\label{eq: H_decomposition}
\H=\A\odot \p 
\end{equation}
where $\p\in \R^{\dimz}$ contains line susceptances, $\odot$ denotes row-wise element-wise multiplication, and $\A\in \R^{\dimz\times \dimx}$ is the topology incidence matrix:
\begin{equation}
\label{eq: incidence_matrix}
\A_{\busi\busj}=
\left\{
\begin{aligned}
&1\quad \text{if }\v_\busj\text{ is the initial node of link }\e_\busi, \\
&-1 \quad \text{if }\v_\busj\text{ is the terminal node of link }\e_\busi, \\
&0 \quad \text{otherwise}.
\end{aligned}
\right.
\end{equation}
The weighted least-squares state estimate is tested by the standard residual sensitivity matrix
\begin{equation}
\label{eq:bdd_residual_matrix}
\S = \mathbf{I} - \H(\H^\top \eR^{-1} \H)^{-1} \H^{\top}\eR^{-1},
\end{equation}
where $\eR$ is the measurement-noise covariance.
In the noiseless case, $\S_o = \mathbf{I} - \H(\H^\top \H)^{-1}\H^\top$ is the orthogonal projector onto $\nul(\H^T)$.
Hence an injected vector $\z_a$ evades the residual test whenever $\z_a\in\row(\H^T)$.

$\H$ has rank $\dimx-1$ for a connected grid after the reference angle is fixed, so $\operatorname{dim}\nul(\H^T)=q$.

\paragraph{Blind \ac{fdia}}

A blind \ac{fdia} attacker observes historical branch-flow measurements but does not know the system information.
The attacker is successful if it can construct a nonzero vector $\z_a$ such that the attacked measurement $\z+\z_a$ remains below the \ac{bdd} threshold while changing the estimated state.
In the noiseless DC model, this means constructing $\z_a\in\row(\H^T)$ without direct access to $\H$.
Existing blind methods estimate this row space statistically, for example, through subspace estimation, \ac{pca}-type factorization, or partial topology-assisted learning \cite{kim2014subspace,yu2015blind,yang2022blind,lakshminarayana2020data,higgins2021topology}.

\subsection{Cycle space and 2-isomorphism}

$\H$ can be regarded as a weighted incidence matrix of $\A$.
The cycle space property of $\A$ has been well studied in graph theory \cite{biggs1997algebraic, godsil2013algebraic, diestel2005cycle}.

\begin{definition}[Cycle terminology]\label{def:cycle}
A cycle is a closed walk with no repeated vertices except that its first and last vertices are the same, represented by its edge set $\c\subseteq\E$.
A spanning tree $S_\G$ is a cycle-free connected subgraph containing all buses.
Edges outside $S_\G$ are chords $\mathcal{C}_\G$.
Each chord together with $S_\G$ induces a fundamental cycle.
\end{definition}

\begin{definition}[Biconnected component]\label{def:biconnected_component}
A biconnected component (also known as a 2-connected component or block) of an undirected graph $\G$ is a maximal subgraph in which the removal of any single node does not disconnect the subgraph.
Equivalently, a biconnected component is a subgraph in which every pair of nodes is connected by at least two internally node-disjoint paths.
In this paper, we denote a biconnected component by $\bi = (\V_\bi, \E_\bi)$, where $\V_\bi$ and $\E_\bi$ are the node and edge sets of the component, respectively.
\end{definition}

\begin{definition}[Cycle space]\label{def:cycle_space}
The cycle space of $\G$ is $\C(\G)= \nul(\A^T) \subseteq\R^{\dimz\times q}$.
It is spanned by signed cycle-indicator vectors $\chi_{\C}\in\R^{\dimz}$, whose entries are $+1$ or $-1$ on oriented cycle edges and $0$ otherwise.
All possible cycle space of $\G$ forms the cycle space set $\cyc$
\end{definition}

There are $q$ linearly independent cycles (that are fundamental cycles).
We denote by $F_\C$ a fundamental cycle set associated with a spanning tree.

\begin{definition}[2-isomorphic graphs]\label{def:2_isomorphic}
Two graphs are strictly 2-isomorphic if there is a one-to-one edge correspondence under which cycles correspond to cycles \cite{whitney19332}.
Equivalently, two graphs are 2-isomorphic iff they have the same cycle space \cite{gross2018graph}.
\end{definition}

\section{Minimum Physical Information for Complete \ac{fdia}}
\label{sec:mini_info}

\subsection{DC weighted \acl{cs}}

For a fixed Jacobian $\H$, define the noiseless stealthy attack space as
\begin{equation}
    \label{eq:complete_attack_space}
    \mathcal{A}(\H)=\row(\H^T).
\end{equation}
An attacker has complete \ac{fdia} capability if it can recover the entire subspace.

\begin{definition}[Weighted cycle space]
\label{def:weighted_cycle_space}
Let $\mathbf{D}_{\p}=\operatorname{diag}(\p)$.  For an oriented cycle
$\c$, its weighted signed indicator is
\begin{equation}
    \label{eq:weighted_signed_indicator}
    \chi_{\c,w}=\mathbf{D}_{\p}^{-1}\chi_{\c}.
\end{equation}
The weighted cycle space is
\begin{equation}
    \C_w(\G)=
    \operatorname{span}\{\chi_{\c,w}:\c\text{ is a cycle of }\G\}.
\end{equation}
\end{definition}

\begin{theorem}[DC weighted cycle space completeness]
\label{theorem:null_H_weighted_cycle}
\label{theorem:minimum_complete_fdia}
Assume that $\G$ is connected, every entry of $\p$ is nonzero, and
$\H=\mathbf{D}_{\p}\A$.  The following statements are equivalent:
\begin{enumerate}
    \item the complete stealthy attack space $\mathcal{A}(\H)$ is known;
    \item the weighted cycle space $\C_w(\G)$ is known;
    \item a cycle basis and block-consistent relative line parameters on its cycle edges are known.
\end{enumerate}
In particular,
\begin{equation}
    \label{eq:null_H_weighted_cycle}
    \nul(\H^T)=\C_w(\G),
    \qquad
    \mathcal{A}(\H)=\C_w(\G)^\perp .
\end{equation}
For a cycle basis $F_{\C}=\{\c_1,\ldots,\c_q\}$ and any invertible
$\mathbf{R}\in\R^{q\times q}$, let
\begin{equation}
    \label{eq:weighted_basis_invariance}
    \begin{aligned}
    \mathbf{N}_w
    &=[\chi_{\c_1,w},\ldots,\chi_{\c_q,w}],\\
    \operatorname{col}(\mathbf{N}_w\mathbf{R})
    &=\operatorname{col}(\mathbf{N}_w)=\C_w(\G).
    \end{aligned}
\end{equation}
The attacker requires $\operatorname{col}(\mathbf{N}_w)$.
$\mathbf{N}_w$ and $\mathbf{N}_w\mathbf{R}$ contain identical information.
\end{theorem}

Minimum refers to the necessary and sufficient conditions for recovering the entire subspace.

\begin{proof}
For every cycle $\c$, $\A^T\chi_{\c}=\mathbf{0}$.  Hence
\begin{equation}
    \label{eq:cycle_kernel}
    \H^T\chi_{\c,w}
    =\A^T\mathbf{D}_{\p}\mathbf{D}_{\p}^{-1}\chi_{\c}
    =\mathbf{0}.
\end{equation}
The $q$ fundamental-cycle indicators are independent.
Multiplication by the invertible matrix $\mathbf{D}_{\p}^{-1}$ preserves
independence, while
$\operatorname{dim}\nul(\H^T)=q$.  Therefore they form a basis of
$\nul(\H^T)$, which proves the first equality in
\cref{eq:null_H_weighted_cycle}.  The second follows by taking the
orthogonal complement.  Consequently, knowing $\C_w(\G)$ gives the entire
attack space, and knowing the entire attack space gives $\C_w(\G)$.
This proves the equivalence of the first two statements.

On each cycle, the nonzero entries of its weighted indicator determine the relative line parameters.  Shared branches align these ratios within each
biconnected component.
Thus, the weighted cycle space gives the third
statement. 
Conversely, the cycle basis and block-consistent relative parameters form the weighted indicators in
\cref{eq:weighted_signed_indicator}.
Their span is $\C_w(\G)$.  
Therefore, the third statement also gives the second.  
Finally, right multiplication by an invertible matrix changes only the basis coordinates, proving
\cref{eq:weighted_basis_invariance}.
\end{proof}

\begin{corollary}[Identifiability]
\label{corollary:Identifiability_DC}
Complete \ac{fdia} does not require distinguishing graphs within a 2-isomorphism class, identifying the parameter scale of a biconnected component, or identifying bridge parameters.
\end{corollary}

\begin{proof}
$\C_w(\G)$ determines the unweighted cycle
space and therefore the topology only up to 2-isomorphism.
Within each biconnected component $\bi$, the weighted cycle vectors determine relative
line parameters, but not their scale:
\begin{equation}
    \label{eq:bi_component_parameter}
    \p_{\bi}=\const_{\bi}\hat{\p}_{\bi}.
\end{equation}
This scale does not change $\C_w(\G)$. 
A bridge is in no cycle, so its parameter does not appear in the weighted cycle space and is not needed for complete \ac{fdia}.
\end{proof}

The cycle basis is also not unique. 
Every valid basis spans the same weighted cycle space in the noiseless case, as shown by
\cref{eq:weighted_basis_invariance}.

\subsection{AC \acl{cm}}

Let $s_e=P_e+\mathrm{j}Q_e$ be the complex power stored at the from end of
branch $e=(i,j)$, and let $u_i=|V_i|^2$.  With the branch current $I_e=Y_{ff,e}V_i+Y_{ft,e}V_j$, the conjugate branch-power equation gives
\begin{equation}
\overline{s_e}=Y_{ff,e}u_i+Y_{ft,e}\overline{V_i}V_j.
\end{equation}
where $Y_{ff,e}$ is the branch's from-end self-admittance.
$Y_{ft,e}$ is the transformer or mutual admittance from the to-end voltage $V_j$ to the from-end current $I_e$.
Hence, when $V_i\neq0$ and $Y_{ft,e}\neq0$, the branch voltage ratio is

\begin{align}
    \label{eq:ac_branch_ratio}
    \rho_e\equiv\frac{V_j}{V_i}
    =A_e+B_e\frac{\overline{s_e}}{u_i},
    \qquad\\ \notag
    A_e=-\frac{Y_{ff,e}}{Y_{ft,e}},\quad
    B_e=\frac{1}{Y_{ft,e}}.
\end{align}

Let $\mathbf C\in\{-1,0,1\}^{q\times\dimz}$ be an oriented fundamental-cycle matrix. 
Voltage ratios telescope around each cycle,
so every feasible AC measurement satisfies
\begin{equation}
    \label{eq:ac_cycle_holonomy}
    \mathcal H_k
    \equiv\prod_{e=1}^{\dimz}\rho_e^{C_{ke}}-1=0,
    \qquad k=1,\ldots,q.
\end{equation}
Each equation in \cref{eq:ac_cycle_holonomy} is complex and therefore gives
two real constraints.

\begin{theorem}[AC cycle manifold completeness]
\label{theorem:ac_cycle_manifold}
Fix the reference-bus voltage magnitude and angle.
Suppose that the grid is connected, the bus voltages and $Y_{ft,e}$ are nonzero, and the full AC measurement Jacobian has rank $2\dimx-2$, and the $P/Q$ Jacobian on every spanning tree is nonsingular. 
The noiseless measurement set is a manifold of dimension $2\dimx-2$.
The chord measurements obey a map:
\begin{equation}
    \label{eq:ac_completion_map}
    \mathbf z_C=F_T(\mathbf z_T),
    \qquad \mathbf z_C\in\R^{2q}.
\end{equation}
Zero AC residual, \cref{eq:ac_completion_map}, and the
$q$ complex cycle equations in \cref{eq:ac_cycle_holonomy} are equivalent.
\end{theorem}

\begin{proof}
Let $\h_T$ and $\h_C$ denote the tree and chord parts map between the AC measurement and state. 
The Jacobian of
$\h_T:\R^{2\dimx-2}\rightarrow\R^{2\dimx-2}$ is nonsingular.
The tree measurements are sufficient to determine the voltage state $\x=g_T(\mathbf z_T)$.
and
\[
F_T(\mathbf z_T)=\h_C\!\left(g_T(\mathbf z_T)\right)
\] is a chord completion map.

Necessity follows because $\rho_e=V_j/V_i$, the voltage ratios multiply by one around every closed cycle.

For sufficiency, consider measurements $(\mathbf z_T,\mathbf z_C)$ that satisfy all fundamental-cycle equations.
First recover $\x=g_T(\mathbf z_T)$.
The tree measurements then give the tree-edge voltage ratios of this state.
Each fundamental cycle contains one chord $c=(i,j)$ and the unique tree path between buses $i$ and $j$.  
\Cref{eq:ac_branch_ratio} gives
\[
\overline{s_c}
=u_i\frac{\rho_c-A_c}{B_c}.
\]
Because $u_i>0$ and $B_c=1/Y_{ft,c}\neq0$, this ratio determines one unique chord power $s_c$.  
Every chord measurement equals $\h_C(g_T(\mathbf z_T))$, so $\mathbf z_C=F_T(\mathbf z_T)$ and the complete measurement vector is produced by one AC state. 
This proves the equivalence with zero AC residual.
\end{proof}

\begin{corollary}[Identifiability]
\label{corollary:Identifiability_AC}
Under the assumptions of \cref{theorem:ac_cycle_manifold},
The branch parameters must be consistent across cycles in the same
biconnected component.
Their voltage-scaling ambiguity and alignment across different components cannot be identified.
Bridge parameters are also not separately identifiable or
required.
\end{corollary}

\begin{proof}
Every cycle is contained in one biconnected component, while a bridge is in no cycle.  
The rows of $\mathbf C$ can be grouped by biconnected component, and $C_{ke}=0$ for every bridge $e$.
\Cref{eq:ac_cycle_holonomy} consequently requires branch parameters only among cycles in the same component and places no cycle constraint on a bridge parameter.

The parameters are subject to a voltage-scaling ambiguity.
For any nonzero bus scalars $h_i$, with $h_{\mathrm{ref}}=1$, set
\begin{equation}
    W_i=h_iV_i,\qquad
    A'_e=\frac{h_j}{h_i}A_e,\qquad
    B'_e=h_j\overline{h_i}B_e,
    \quad e=(i,j).
    \label{eq:ac_voltage_gauge}
\end{equation}
Since $|W_i|^2=|h_i|^2u_i$,
\begin{equation}
    A'_e+B'_e\frac{\overline{s_e}}{|W_i|^2}
    =\frac{h_j}{h_i}
      \left(A_e+B_e\frac{\overline{s_e}}{u_i}\right)
    =\frac{W_j}{W_i}.
\end{equation}
The factors $h_j/h_i$ cancel around every cycle.
This transformation changes the coefficient representation, including parameters on bridges and between components, without changing the cycle manifold.
By \cref{theorem:ac_cycle_manifold}, an attacked measurement
$\mathbf z+\mathbf z_a$ has zero
AC residual if and only if, for some sufficiently small
$\Delta\mathbf z_T$,
\begin{equation}
    \label{eq:ac_parameters_attack}
    \mathbf z_a=
    \begin{bmatrix}
        \Delta\mathbf z_T\\
        F_T(\mathbf z_T+\Delta\mathbf z_T)-F_T(\mathbf z_T)
    \end{bmatrix}.
\end{equation}
For $\Delta\mathbf z_T\neq\mathbf0$, the parameters inverse $g_T$ gives a different voltage state. 
The blockwise cycle manifold is the identifiable object required for
complete \ac{fdia}.
\end{proof}

The DC result in \cref{corollary:Identifiability_DC} is the linear example:
2-isomorphic graphs, biconnected-component parameter scales, and bridge
parameters are indistinguishable because they leave the weighted cycle space,
and therefore the complete attack space, unchanged.

\subsection{Connection between AC \acl{cm} and DC \acl{cs}}

Under the flat voltage, small angle assumption, $\rho_e \approx 1-\mathrm {j}x_eP_e$. 
Taking the first-order logarithm of \cref{eq:ac_cycle_holonomy} yields
\begin{equation}
    \label{eq:ac_cycle_linearization}
    \mathbf C\mathbf D_{\mathbf x}\,\delta\mathbf P=\mathbf0,
    \qquad
    \mathbf D_{\mathbf x}
    =\operatorname{diag}(\mathbf x)=\mathbf D_{\p}^{-1}.
\end{equation}
Let $\mathcal M_{\mathrm{AC}}^{P}$ denote this fixed-voltage active power of the AC cycle manifold, with the flat point represented by $\mathbf P=\mathbf0$. 
The columns of $\mathbf D_{\mathbf x}\mathbf C^T$ are the weighted cycle indicators in \cref{eq:weighted_signed_indicator}. 
Its tangent and normal spaces are
\begin{equation}
    \label{eq:ac_dc_tangent_normal}
    \begin{aligned}
    T_{\mathbf0}\mathcal M_{\mathrm{AC}}^{P}
        &=\nul(\mathbf C\mathbf D_{\mathbf x})
          =\C_w(\G)^\perp
          =\mathcal A(\H),\\
    N_{\mathbf0}\mathcal M_{\mathrm{AC}}^{P}
        &=\operatorname{col}(\mathbf D_{\mathbf x}\mathbf C^T)
          =\C_w(\G).
    \end{aligned}
\end{equation}
The DC \acl{cs} is the active-power tangent space of the AC cycle manifold.
The DC cycle-space model is the first-order approximation of the AC cycle manifold.

A blind \ac{fdia} attacker only has information about the measurement.
Learn \ac{cs} and \ac{cm} from measurement is a physically informed machine learning approach.
In the next two sections, we are going to propose the blind \ac{fdia} realization method based on \ac{cs} and \ac{cm}.
There are two stages: cycle-based estimation and parameter estimation.
Cycle-based estimation is a combinatorial problem.
Solving such a problem brutally is NP-hard and requires almost infinite time to optimize on current electronic computers.
We are going to realize two blind \ac{fdia}:
Computationally Constrained Blind \ac{fdia} Realization, which requires only the measurement information.
Computationally Unconstrained Blind \ac{fdia} Realization, which requires the measurement and topology information.

\section{Computationally Constrained Blind \ac{fdia} Realization}
In this section, the given information is measurement-only, which is common in classic blind \ac{fdia}.

\subsection{DC realization}

\label{sec:measurement_identification}
Let
\begin{equation}
    \Z=[\z_1,\ldots,\z_\time]=\H\X+\noiseE
    \in\R^{\dimz\times\time}
    \label{eq:unconstrained_measurements}
\end{equation}
contain all branch-flow measurement channels.
We assume $\time>\dimz$. 
Let $r=\dimx-1$.  
The output is a full-column-rank weighted cycle matrix
\begin{equation}
    \label{eq:realization_objective}
    \widehat{\mathbf N}_{\C}
    =[\widehat{\n}_1,\ldots,\widehat{\n}_q]
    \in\R^{\dimz\times q},
    \qquad
    \widehat{\mathbf N}_{\C}^{T}\Z\approx\mathbf0 .
\end{equation}
Its column space estimates $\nul(\H^T)$, while
$\nul(\widehat{\mathbf N}_{\C}^{T})$ estimates the physical measurement
and attack space $\row(\H^T)$.  The implementation uses \ac{rms}
normalization, a rank-$r$ \ac{svd}, pivoted QR, sparse support
screening, and \ac{tls} fitting.

\begin{figure*}[t]
    \centering
    \begin{tikzpicture}[
        font=\footnotesize,
        stage/.style={
            draw,
            rounded corners=2pt,
            align=center,
            minimum height=10mm,
            text width=0.185\textwidth,
            inner sep=3pt,
            fill=blue!5
        },
        output/.style={stage,fill=green!8},
        auxiliary/.style={stage,dashed,fill=orange!8},
        flow/.style={-{Latex[length=2.2mm]},thick},
        auxflow/.style={flow,dashed}
    ]
        \node[stage] (data)
        {Branch measurements\\$\Z\in\R^{\dimz\times\time}$};
        \node[stage,right=5mm of data] (rms)
        {\ac{rms} normalization\\$\widetilde{\Z}=\mathbf D_s^{-1}\Z$};
        \node[stage,right=5mm of rms] (svd)
        {Rank-$r$ \ac{svd}\\signal basis $\mathbf U_r$};
        \node[stage,right=5mm of svd] (tree)
        {Pivoted QR\\algebraic tree $\widehat{\S}_{\G}$};

        \node[stage,below=7mm of tree] (screen)
        {Each chord $\e$\\Lasso--\ac{bic} screen $K_{\e}$};
        \node[stage,left=5mm of screen] (bic)
        {\ac{tls} coefficient ordering\\and \ac{bic} support $S_{\e}^{(0)}$};
        \node[stage,left=5mm of bic] (refine)
        {$10^{-2}$ support refinement\\and \ac{tls} refit};
        \node[output,left=5mm of refine] (cycles)
        {Weighted cycle matrix\\$\widehat{\mathbf N}_{\C}$, rank $q$};

        \node[output,below=7mm of cycles] (attack)
        {Recovered attack space\\
        $\nul(\widehat{\mathbf N}_{\C}^{T})$};
        \node[auxiliary,below=7mm of bic] (parameters)
        {Auxiliary overlap alignment\\relative parameters $\widehat{\p}$};

        \draw[flow] (data) -- (rms);
        \draw[flow] (rms) -- (svd);
        \draw[flow] (svd) -- (tree);
        \draw[flow] (tree) -- (screen);
        \draw[flow] (screen) -- (bic);
        \draw[flow] (bic) -- (refine);
        \draw[flow] (refine) -- (cycles);
        \draw[flow] (cycles) -- (attack);
        \draw[auxflow] (refine.south) |- (parameters.west);
        \node[font=\scriptsize,anchor=south west]
        at ([xshift=1mm]parameters.north west) {auxiliary output};
    \end{tikzpicture}
    \caption{Measurement-only cycle-space recovery pipeline.}
    \label{fig:pipeline}
\end{figure*}
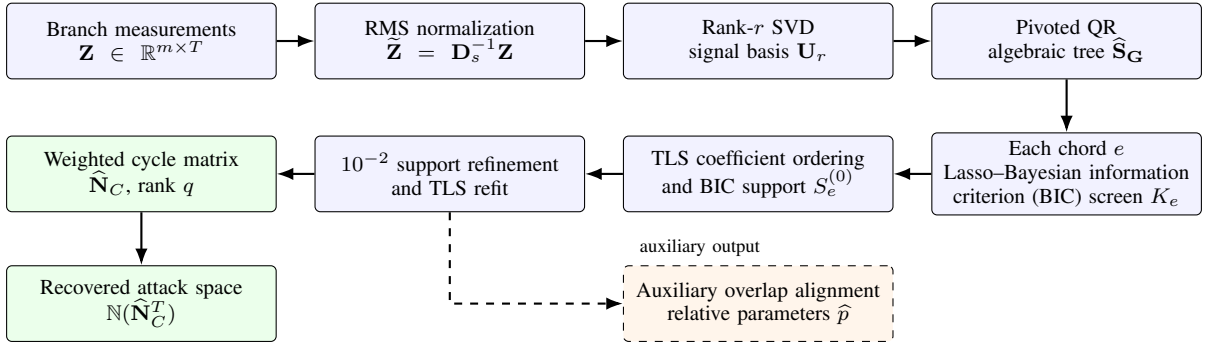

\subsubsection{\ac{rms} normalization and algebraic-tree recovery}

For each branch $\e$, we define its uncentered \ac{rms} scale by
\begin{align}
    \label{eq:rms_normalization}
    s_{\e}=\left(\frac{1}{\time}\|\Z_{\e,:}\|_2^2\right)^{1/2},
    \\ \notag
    \mathbf D_s=\operatorname{diag}(s),
    \qquad
    \widetilde{\Z}=\mathbf D_s^{-1}\Z .
\end{align}
where all $s_{\e}$ must be nonzero.
\Cref{eq:rms_normalization} is a branchwise scale correction.

We decompose
$\widetilde{\Z}=\mathbf U\mathbf\Sigma\mathbf V^T$ such that
$\mathbf U_r\in\R^{\dimz\times r}$ contain the first $r$ left singular
vectors.
Apply column-pivoted QR to the branch coordinates:
\begin{equation}
    \label{eq:QR_pivot}
    \mathbf U_r^T\mathbf P_{\pi}=\mathbf Q\mathbf R,
    \qquad
    \widehat{\S}_{\G}
    =\{\e_{\pi(1)},\ldots,\e_{\pi(r)}\}.
\end{equation}
The remaining set
$\widehat{Q}=\E\setminus\widehat{\S}_{\G}$ contains $q$ candidate chords.
The implementation processes these chords in increasing branch-index order.
In noiseless data, the selected $r$ independent branches form a spanning tree.  
In noisy data, Column-pivoted QR selects the most independent branches.

\subsubsection{Cycle basis recovery}
For each $\e\in\widehat Q$, the normalized chord measurement is
regressed on all normalized tree measurements:
\begin{equation}
    \label{eq:lasso_screening}
    \widetilde{\Z}_{\e,:}^{T}
    \approx
    \widetilde{\Z}_{\widehat{\S}_{\G},:}^{T}\boldsymbol\beta_{\e}.
\end{equation}
We use a no-intercept Lasso path and select its point by the \ac{bic}.
The nonzero entries of $\boldsymbol\beta_{\e}$ define the screened tree
set $K_{\e}\subseteq\widehat{\S}_{\G}$.  
If fewer than two branches are selected, ordinary least squares is used, and the two largest parameters are retained.  
This guarantees a chord candidate with at least two tree branches.

On $S_{\e}^{\mathrm{scr}}=\{\e\}\cup K_{\e}$, we compute the unit \ac{tls} vector
\begin{equation}
    \label{eq:screened_tls}
    \mathbf w_{\e}^{\mathrm{scr}}
    \in\underset{\|\mathbf w\|_2=1}{\arg\min}\;
    \|\mathbf w^T\widetilde{\Z}_{S_{\e}^{\mathrm{scr}},:}\|_2^2 .
\end{equation}
The branches in $K_{\e}$ are ordered by decreasing corresponding magnitude
in $\mathbf w_{\e}^{\mathrm{scr}}$.
Equal magnitudes are ordered by branch index.
Let this order be $k_{\e,1},\ldots,k_{\e,|K_{\e}|}$ and define
\begin{equation}
    \label{eq:nested_cycle_candidates}
    C_{\e,j}=\{\e,k_{\e,1},\ldots,k_{\e,j}\},
    \qquad j=2,\ldots,|K_{\e}|.
\end{equation}
For each support, the \ac{tls} residual and explicit \ac{bic} score are
\begin{align}
    \label{eq:tls_bic}
    \rho_{\e,j}
    &=\sigma_{\min}^2(\widetilde{\Z}_{C_{\e,j},:}),\notag\\
    \operatorname{BIC}_{\e,j}
    &=\time\log\!\left[
      \max\!\left(\frac{\rho_{\e,j}}{\time},
      \varepsilon_{\mathrm{tiny}}\right)\right]
      +|C_{\e,j}|\log\time,
\end{align}
where $\varepsilon_{\mathrm{tiny}}$ is the machine positive-normal threshold.

The preliminary cycle $S_{\e}^{(0)}$ is the candidate with the smallest score. 
Thus, the \ac{bic} is used twice: first to screen the Lasso path and then to
choose the \ac{tls} prefix length.

\subsubsection{Weighted-constraint fitting and support refinement}
For any selected cycle $S$, let
\begin{equation}
    \label{eq:practical_parameter_optimization}
    \mathbf w_{\e}(S)
    \in\underset{\|\mathbf w\|_2=1}{\arg\min}\;
    \|\mathbf w^T\widetilde{\Z}_{S,:}\|_2^2 .
\end{equation}
This is the smallest left singular vector of
$\widetilde{\Z}_{S,:}$.  To return to the original branch-flow coordinates,
embed it in $\R^{\dimz}$ as
\begin{align}
    \label{eq:raw_cycle_constraint}
    \overline{\n}_{\e}(S)[\ell]
    =
    \begin{cases}
      \mathbf w_{\e}(S)[\ell]/s_{\ell},&\ell\in S,\\
      0,&\ell\notin S,
    \end{cases}
    \\ \notag
    \n_{\e}(S)=
    \frac{\overline{\n}_{\e}(S)}
         {\|\overline{\n}_{\e}(S)\|_2}.
\end{align}
The sign is fixed by making the parameter on the smallest-index cycle branch non-negative.
This convention does not change the cycle equation.
Stacking $\n_{\e}(S_{\e}^{(0)})$ gives the preliminary matrix
$\widehat{\mathbf N}_{\C}^{(0)}$.

We then apply the fixed raw-coordinate threshold
$\tau=10^{-2}$:
\begin{equation}
    \label{eq:support_refinement}
    \widehat{\c}_{\e}
    =\{\ell\in\E:
      |\widehat{\n}_{\e}^{(0)}(\ell)|>\tau\}.
\end{equation}
Equation~\eqref{eq:practical_parameter_optimization} is solved again on each
$\widehat{\c}_{\e}$, and \cref{eq:raw_cycle_constraint} gives the final
$\widehat{\n}_{\e}$.  We define
$\widehat F_{\C}=\{\widehat{\c}_{\e}:\e\in\widehat Q\}$ such that the final matrix
is accepted only if
\begin{equation}
    \label{eq:constraint_rank_check}
    \rank(\widehat{\mathbf N}_{\C})=q.
\end{equation}
The threshold is applied after the raw-coordinate conversion and unit normalization.
The chord is not forced to remain, and a refined cycle can contain only two branches.
Hence \cref{eq:constraint_rank_check} guarantees independent constraints. 
Also,
\begin{equation}
    \label{eq:estimated_bridge_set}
    \widehat{\E}_{\mathrm{no\mbox{-}cycle}}
    =\E\setminus\bigcup_{\e\in\widehat Q}\widehat{\c}_{\e}
\end{equation}
is the set of branches absent from every recovered cycle.
It equals the graph-bridge set when the recovered cycle basis link set is correct.

\subsubsection{Parameter estimation}
If $\widehat{\c}$ is a true cycle, then its exact weighted constraint obeys
\begin{equation}
    \label{eq:cycle_parameter_relation}
    \n_{\widehat{\c}}
    =\beta_{\widehat{\c}}
      \mathbf D_{\p_{\widehat{\c}}}^{-1}
      \chi_{\widehat{\c}},
    \qquad \beta_{\widehat{\c}}\neq0.
\end{equation}
Thus, the cycle gives the branches, the signs give an orientation up to one global sign, and the magnitudes give relative line parameters:
\begin{equation}
    \label{eq:param_est}
    \frac{\p(\ell_1)}{\p(\ell_2)}
    =\frac{|\n_{\widehat{\c}}(\ell_2)|}
           {|\n_{\widehat{\c}}(\ell_1)|},
    \qquad \ell_1,\ell_2\in\widehat{\c}.
\end{equation}

We also return a separate parameter estimate.  It groups
recovered cycles by shared branches and estimates a new smallest-eigenvector
constraint $\check{\n}_{\c}$ from
$\Z_{\c,:}\Z_{\c,:}^{T}$ for each cycle.  If cycles $\c_i$ and $\c_j$ share
$J_{ij}\neq\emptyset$, their coefficient scales are aligned by
\begin{equation}
    \label{eq:cycle_alignment}
    \gamma_{ij}
    =\frac{1}{|J_{ij}|}
      \sum_{\ell\in J_{ij}}
      \frac{|\check{\n}_{\c_i}(\ell)|}
           {|\check{\n}_{\c_j}(\ell)|},
    \qquad
    \check{\n}_{\c_j}\leftarrow
    \gamma_{ij}\check{\n}_{\c_j}.
\end{equation}
After propagating these ratios through each bi-connected component, the inverse coefficient magnitudes yield $\widehat{\p}$ up to a scale factor.  

\begin{algorithm}[!t]
\SetAlgoLined
\KwIn{$\Z$, branch set $\E$, number of buses $\dimx$}
\KwOut{$\widehat{\S}_{\G}$, $\widehat F_{\C}$,
$\widehat{\mathbf N}_{\C}$, auxiliary $\widehat{\p}$}
Compute $\widetilde{\Z}$ using \cref{eq:rms_normalization}\;
Compute $\widehat{\S}_{\G}$ by rank-$r$ \ac{svd} and pivoted QR
using \cref{eq:QR_pivot}\;
\For{$\e\in\widehat Q$}
{
    Screen $K_{\e}$ by no-intercept Lasso--\ac{bic} using
    \cref{eq:lasso_screening}\;
    If $|K_{\e}|<2$, retain the two largest ordinary least-squares parameters\;
    Order $K_{\e}$ by the \ac{tls} vector in \cref{eq:screened_tls}\;
    Select $S_{\e}^{(0)}$ by \cref{eq:nested_cycle_candidates,eq:tls_bic}\;
    Fit the preliminary raw-coordinate constraint by
    \cref{eq:practical_parameter_optimization,eq:raw_cycle_constraint}\;
}
Refine every support by \cref{eq:support_refinement} and refit its
constraint\;
Verify \cref{eq:constraint_rank_check}\;
Estimate auxiliary relative parameters using \cref{eq:cycle_alignment}\;
\KwRet{$\widehat{\S}_{\G},\widehat F_{\C},
\widehat{\mathbf N}_{\C},\widehat{\p}$}\;
\caption{\label{alg:tree_bi_detection}
Measurement-only cycle-space realization}
\end{algorithm}

\subsubsection{Recovered attack space and residual}
We define:
\begin{equation}
    \label{eq:practical_cycle_projector}
    \begin{aligned}
    \mathcal W_{\mathrm{pr}}
    &=\operatorname{col}(\widehat{\mathbf N}_{\C}),
    &\P_{\mathrm{pr}}
    &=\widehat{\mathbf N}_{\C}
      \widehat{\mathbf N}_{\C}^{+},\\
    \widehat{\mathcal A}
    &=\nul(\widehat{\mathbf N}_{\C}^{T})
    =\mathcal W_{\mathrm{pr}}^{\perp}.&&
    \end{aligned}
\end{equation}
For any $\mathbf u\notin\mathcal W_{\mathrm{pr}}$ and $\alpha>0$, a recovered
attack direction is
\begin{equation}
    \label{eq:practical_attack}
    \z_a=\alpha
    \frac{(\I-\P_{\mathrm{pr}})\mathbf u}
         {\|(\I-\P_{\mathrm{pr}})\mathbf u\|_2},
\end{equation}
such that $\widehat{\mathbf N}_{\C}^{T}\z_a=\mathbf0$.  

For test data $\Z_{\mathrm{test}}$, we report the raw cycle
residual
\begin{equation}
    \label{eq:raw_cycle_residual}
    \mathbf E_{\C}
    =\widehat{\mathbf N}_{\C}^{T}\Z_{\mathrm{test}}.
\end{equation}
Since the unit-norm columns are not generally orthogonal, the norm of \cref{eq:raw_cycle_residual} also depends on the basis of conditioning.

The realization uses no grid topology, incidence matrix, Jacobian, line parameters, or noise covariance. 
Its main operations are one rank-$r$ \ac{svd}, one \ac{rrqr}, $q$ Lasso paths, and at most $q(r-1)$ small \ac{tls} evaluations. 
The main failure comprises three-coordinate errors, omitted Lasso branches, the fixed $10^{-2}$ threshold, removal of the chord during refinement, and independent final-cycle fits that need not share consistent line parameters.

\subsection{AC limitation}

Because the AC \ac{cm} is nonlinear, we can not find a shortcut to extract the cycle basis under AC. 

\section{Computationally Unconstrained Blind \ac{fdia} Realization}
In this section, the given information is measurement plus topology.

\subsection{DC Realization}
\label{sec:unconstrained_ceiling}

There are $q$ fundamental cycles for a full rank null space estimation.
The target to optimize is:
\begin{align}
    \label{eq:opt_dc_un}
    \min_{\hat{\mathbf{N}}_\C} \|\hat{\mathbf{N}}_\C^T\Z\|^2 \\ s.t. \|\hat{\mathbf{N}}_\C\|^2= 1, \C\in \cyc.
\end{align}
This problem was solved in \cite{li2026cycle}, where proves the minimum cycle basis achieve the optimal generalization error.
The minimum cycle basis realization of DC unconstrained is the upper bound of DC blind \ac{fdia}.

\subsection{AC realization}

\subsubsection{Gauged Tree–Chord Cycle Manifold Fitting}
The AC realization assumes that the grid topology, the direction of each branch measurement, and the reference bus are known, while branch admittances and voltage states are not supplied.
A spanning tree $S$ is chosen to fix the fundamental cycles.
Let $l$ be the number of bridge branches in the graph.
According to \cref{eq:ac_branch_ratio}, directly optimizing the AC manifold requires to estimate $A_e,B_e$ and the voltage magnitude $u_i$.
Directly estimating the voltage magnitude is unacceptable because it contributes $(n-1)T$ parameters.
We proposed Gauged Tree–Chord Cycle Manifold Fitting (GTCM) method, which  only estimate a variant of $B_e$.
$A_e$ is eliminated using gauge method.
An equivalent variant of voltage magnitude is computed during the optimization, instead of working as optimization variables.

We introduce a gauged voltage at every bus:
\begin{equation}
\label{eq:gauged_voltage}
    V_i'=h_iV_i.
\end{equation}
The branch power equation becomes:
\begin{equation}
    \frac{V_j'}{V_i'}=\frac{h_j'}{h_i'}A_e+h_j'h_i'B_e\frac{\overline{s_e}}{|V'_i|^2}.
\end{equation}
We define the variants of $A_e,B_e$ as $A_e',B_e'$:
\begin{equation}
    A_e'= \frac{h_j}{h_i}A_e, B_e'=h_jh_iB_e.
\end{equation}
By imposing the constraint $\frac{h_j}{h_i}A_e=1$, $A_e'=1$ holds everywhere for branches on a spanning tree $S$ because there are $n-1$ free $A_e$ parameters on a tree and there is a unique path for each branch to the reference branch.

Then we can compute the gauged voltage magnitude from the reference branch to every tree branch.
For a tree branch $e=(i,j)$ traversed in the direction, $V_i'$ is already known, so:
\begin{equation}
    V_j'=1+B'_e\frac{\overline{s_e}}{V_i'}
\end{equation}
If a branch orientation is opposite to the tree traversal, the implementation solves a quadratic equation for the unknown sending-end $V_i$ and chooses the high-voltage root:
\begin{equation}
    V_j'=\frac{V_i'+\sqrt{V_i'-4B_e'\overline{s_e}}}{2}
\end{equation}
We can compute all  branch gauged voltage magnitude following the tree traversal.

Now take the chord.
The tree ratios in this cycle are already known.
According to the cycle manifold constraint in \cref{eq:ac_cycle_holonomy}, for a chord $c=(i,j)\in\mathcal{C}$ we define:
\begin{equation}
    S_c=\prod_{e\in S}\rho_e^{C_{ke}}.
\end{equation}
The chord ratio is therefore:
\begin{equation}
    \rho_c= S_c^{-1}
\end{equation}
The chord power flow is:
\begin{equation}
    \label{eq:chord_meas}
    \overline{s_c}=\frac{|V_i|^2}{K_c}(\rho_c-A_c).
\end{equation}
Define:
\begin{equation}
    d_c=\frac{|V_i|^2}{K_c}, \quad g_c=d_c\rho_c.
\end{equation}
Then:
\begin{equation}
    \overline{s_c}=g_c-d_cA_c.
\end{equation}
Suppose we have $\time$ training samples.
For every chord $c$, we can get the closed-form solution of $A_c$:
\begin{equation}
    \hat{A}_c=\frac{\sum_{r=1}^{\time}d_c^{(r)}(g_c^{(r)}-\overline{s_c}^{(r)})}{\sum_{r=1}^{\time}\|d_c^{(r)}\|^2}
\end{equation}
Then we project the power back:
\begin{equation}
    \hat{s}_c^{(r)}= g_c^{(r)}-d_c^{(r)}\hat{A}_c.
\end{equation}
The optimization is established between the predict and measurement chord power:
\begin{equation}
    \min_{B'}\frac{1}{2}\sum_{r=1}^\time\sum_{c\in \mathcal{C}}\|\frac{\hat{s}_c^{(r)}-\overline{s}_c^{(r)}}{\sigma_c}\|^2
\end{equation}

\subsubsection{Attack generation}
In this section, we use the learned parameters to construct the attack. 
We first generate zero-mean Gaussian noise for the tree measurement $s^a_{S}=s_{S}+r\frac{\delta_{S}}{\|\delta_{S}\|^2}$.
Where $\delta_{S}$ is generated under an independent normal distribution.
$r$ is a scaling factor controlling the attack magnitude.
By filling in the learned parameters, we compute the chord measurement $s^a_{\mathcal{C}}$ by following \cref{eq:gauged_voltage} to \cref{eq:gauged_voltage}.
The fit and batched generation use double-precision CUDA computation on a \ac{gpu}.
\section{Defender-Side Implications}
\label{sec:defender_implications}

\subsection{Grid planning}

\subsection{Design Principles for Defending Against Blind \ac{fdia}}
\label{sec:defense_principles}

The detection bound in the \ac{csd} framework \cite{li2026cycle} is optimal precisely because cycle-space recovery is the attacker's minimum requirement.
In this sense, \ac{csd} is dual to the necessity theorem: it measures how far the attacker is from the only structure that can support uniformly stealthy blind \ac{fdia}.

A moving-target defense defeats the attacker if and only if the perturbation changes the recovered 2-isomorphism class or its weighted cycle-space realization.
Perturbations that leave the 2-isomorphism class unchanged may change a nominal topology diagram, but they do not remove the attacker's sufficient information.

Measurement design should therefore restrict the attacker's ability to infer the weighted cycle space.
Meter placement, aggregation, protected channels, and deliberate withholding of selected branch-flow streams are useful to the extent that they deny enough independent cycle information to reconstruct $\nul(\H^T)$.

\section{Experiments}
\label{sec:experiments}

\subsection{Setup}
We evaluate the proposed method on the IEEE 14-, 30-, 57-, and 118-bus test systems. For each grid, the
first $\time_{\mathrm{train}}=\dimz$ operating points form the training set.
The clean branch \ac{rms} values $s_{\e}$ follow \cref{eq:rms_normalization}, and
the noisy training data are
\begin{equation}
    \label{eq:experiment_training_noise}
    \begin{aligned}
    \Z_{\mathrm{train}}^{\mathrm{noisy}}
    &=\Z_{\mathrm{train}}+\noiseE,\\
    \noiseE_{\e,t}&\sim\mathcal N(0,\sigma_{\e}^2),
    &\sigma_{\e}&=0.1\max(s_{\e},10^{-8}).
    \end{aligned}
\end{equation}
In the state-impact experiment, all methods are trained using a single noise
realization generated with random seed 7.  We compare seven DC attack-space
estimates.  \emph{H known} uses $\operatorname{col}(\H)$.
\emph{\ac{cs} unconstrained} uses the true minimum-cycle-basis supports but fits
their weighted parameters from the same noisy data.  \emph{\ac{cs} realization}
uses the measurement-only estimator in
\Cref{sec:measurement_identification}.  \emph{\ac{pca}} estimates the standardized
rank-$r$ measurement space.  \emph{\ac{rmt} perturbation} follows the random-matrix
attack construction of \cite{lakshminarayana2020data}, extended to the 
rank $r$ used in this comparison.  \emph{Matrix reconstruction} implements
the covariance-reconstruction attack of \cite{yang2022blind}.  The
\emph{linear autoencoder} uses a rank-$r$ linear encoder--decoder.

\subsection{State impact versus 95\%-calibrated pass rate}
\label{sec:dc_state_impact}
Let $\mathbf D_{\sigma}\equiv\operatorname{diag}(\sigma)$,
$\H_w=\mathbf D_{\sigma}^{-1}\H$, and let $\mathbf Q_{\perp}$ be an
orthonormal basis of $\nul(\H_w^T)$.  From unattacked noisy trials, we set
\begin{equation}
    \label{eq:dc_bdd_threshold95}
    \tau_{0.95}=Q_{0.95}\{J_i^{(0)}\},\qquad
    J_i=\|\mathbf D_{\sigma}\mathbf Q_{\perp}
    \mathbf Q_{\perp}^T\widetilde{\z}_i\|_2^2,
\end{equation}
where $\widetilde{\z}_i$ is the whitened measurement.  Thus 95\% of nominal
trials pass this fixed threshold.

For a method with attack basis $\mathbf B\in\R^{\dimz\times r}$, reduced QR
gives an orthonormal basis $\mathbf Q_B$ for
$\operatorname{col}(\mathbf D_{\sigma}^{-1}\mathbf B)$.  At strength
$\gamma$, the paired trial direction is
\begin{equation}
    \label{eq:dc_state_impact_attack}
    \widetilde{\z}_{a,i}(\gamma)
    =\gamma\mathbf Q_B
      \frac{\mathbf Q_B^T\mathbf g_i}
           {\|\mathbf Q_B^T\mathbf g_i\|_2},
    \qquad \mathbf g_i\sim\mathcal N(\mathbf0,\I).
\end{equation}
The mean state impact and pass rate are
\begin{equation}
    \label{eq:dc_state_impact_metric}
    \begin{aligned}
    \overline d_{\mathrm{DC}}(\gamma)
    &=\frac{1}{N\sqrt r}\sum_{i=1}^{N}
      \|\H_w^{\dagger}\widetilde{\z}_{a,i}(\gamma)\|_2,\\
    \pi(\gamma)&=\frac{1}{N}\sum_{i=1}^{N}
      \mathbf{1}\{J_i(\gamma)\leq\tau_{0.95}\}.
    \end{aligned}
\end{equation}
Unlike a pass-rate-versus-threshold plot, this metric shows how much the state
changes at the fixed detector operating point and avoids treating a
negligible-impact attack as equally effective.

\begin{figure*}[!t]
    \centering
    \begin{minipage}[t]{0.49\textwidth}
        \centering
        \includegraphics[width=\linewidth]{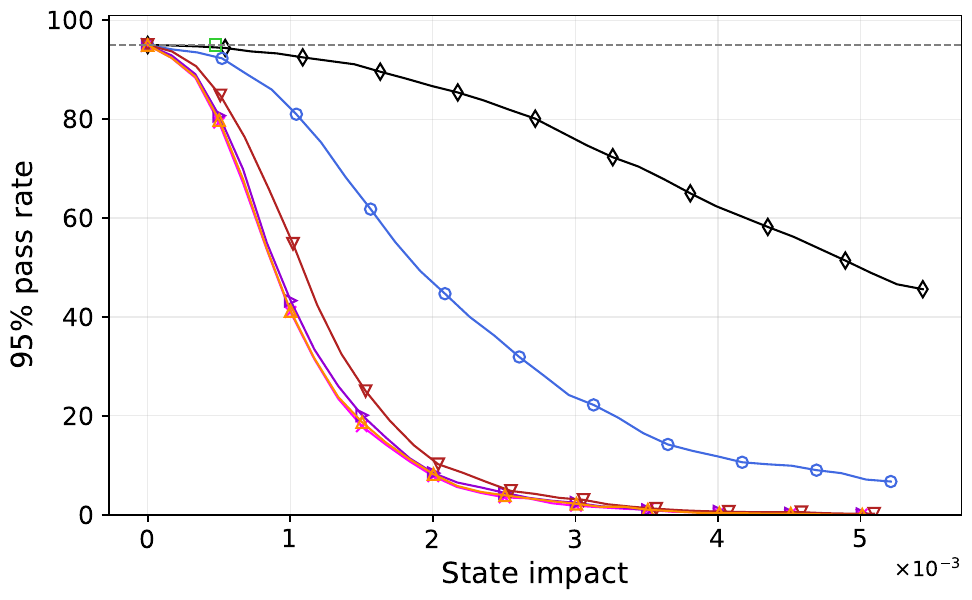}
        \vspace{-0.4em}
        {\footnotesize (a) IEEE 14-bus}
    \end{minipage}\hfill
    \begin{minipage}[t]{0.49\textwidth}
        \centering
        \includegraphics[width=\linewidth]{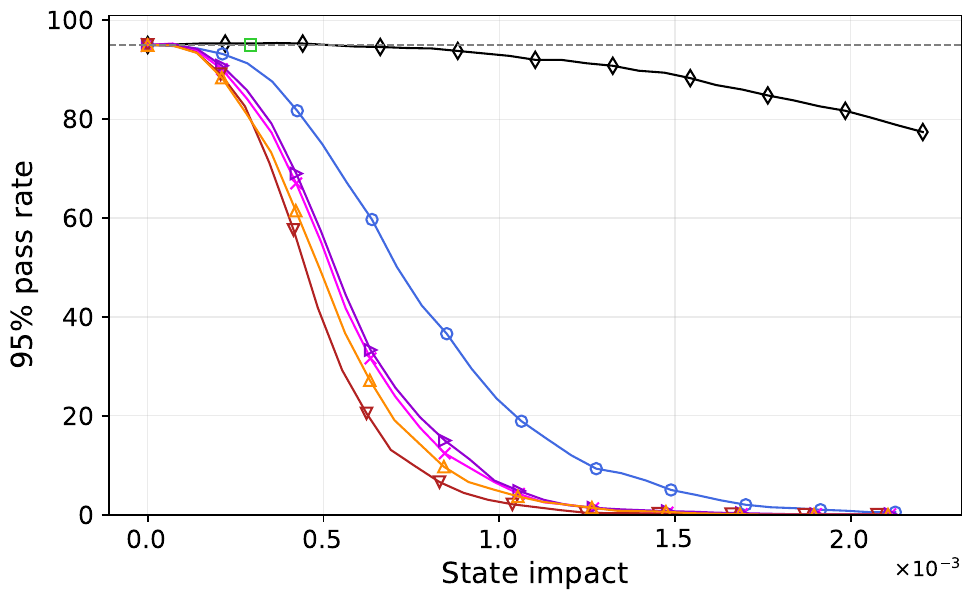}
        \vspace{-0.4em}
        {\footnotesize (b) IEEE 30-bus}
    \end{minipage}

    \vspace{0.3em}
    \begin{minipage}[t]{0.49\textwidth}
        \centering
        \includegraphics[width=\linewidth]{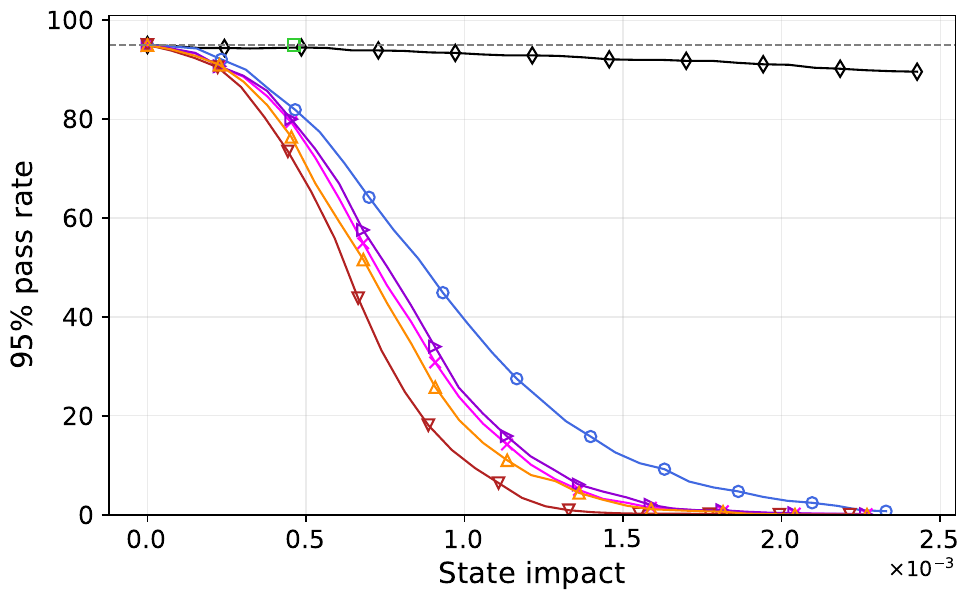}
        \vspace{-0.4em}
        {\footnotesize (c) IEEE 57-bus}
    \end{minipage}\hfill
    \begin{minipage}[t]{0.49\textwidth}
        \centering
        \includegraphics[width=\linewidth]{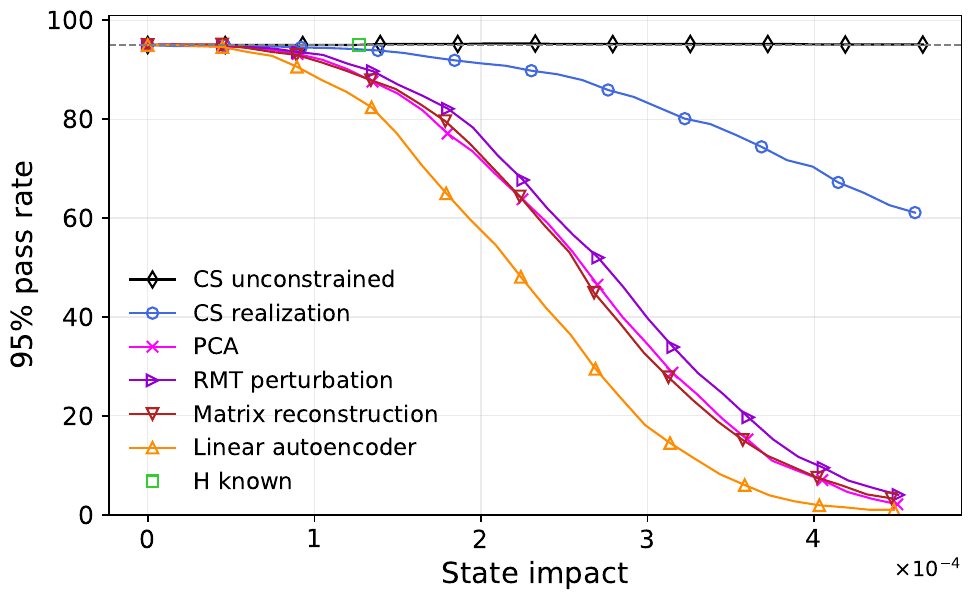}
        \vspace{-0.4em}
        {\footnotesize (d) IEEE 118-bus}
    \end{minipage}
    \caption{DC pass rate at a \ac{bdd} threshold calibrated for 95\% nominal
    acceptance versus mean state impact.  Each curve uses 1,000 paired trials
    and one fixed model learned at 10\% branch-relative training noise.}
    \label{fig:dc_state_impact}
\end{figure*}

In \Cref{fig:dc_state_impact}, \ac{cs} unconstrained remains closest to the nominal
95\% pass rate as state impact grows, showing the value of the correct cycle
supports.  \ac{cs} realization is consistently the strongest measurement-only
realization among the tested methods.  The remaining subspace methods lose
pass rate at smaller state impacts.  The H-known marker is the nominal
noise-only reference, rather than an attack-strength curve.

\subsection{Null-space-overlap noise sensitivity}
This experiment evaluates geometric recovery.
We use IEEE-118, where
\begin{equation}
    \label{eq:ieee118_dimensions}
    \begin{aligned}
    \dimz&=186, &\dimx&=118, &r&=117, &q&=69,\\
    \time_{\mathrm{train}}&=186.&&&&&
    \end{aligned}
\end{equation}
The remaining 1,254 operating points are reserved for the attack experiment.
For 10 logarithmically spaced levels
$\alpha\in[10^{-2},10^{-1}]$, the noisy training matrix in trial $j$ is
\begin{equation}
    \label{eq:noise_sensitivity_design}
    \Z_{\alpha,j}
    =\Z_{\mathrm{train}}+\alpha\mathbf D_s\mathbf\Xi_j,
    \qquad
    \mathbf\Xi_j[\e,t]\sim\mathcal N(0,1).
\end{equation}
There are 20 trials at every level.  Seed 17 is used, and the same
$\mathbf\Xi_j$ is scaled across all levels and passed to all methods.  Every
method is re-estimated for every trial.  A failed \ac{cs} realization trial is
assigned zero overlap, so its curve includes tree, support, coefficient, and
rank failures.

Let $\mathbf Q_0\in\R^{\dimz\times q}$ be an orthonormal basis of
$\nul(\H^T)$ and let $\widehat{\mathbf Q}$ span the estimated detector space.
For the \ac{cs} methods this is
$\operatorname{col}(\widehat{\mathbf N}_{\C})$; for the signal-space
methods it is the orthogonal complement of the estimated rank-$r$ basis.  We
report
\begin{equation}
    \label{eq:null_space_overlap}
    \eta
    =\frac{1}{q}\|\mathbf Q_0^T\widehat{\mathbf Q}\|_F^2
    =\frac{1}{q}\sum_{i=1}^{q}\cos^2\theta_i,
\end{equation}
where $\theta_i$ are the principal angles.  Thus $100\eta$ measures average
subspace alignment.

\begin{figure}[!t]
    \centering
    \includegraphics[width=\linewidth]{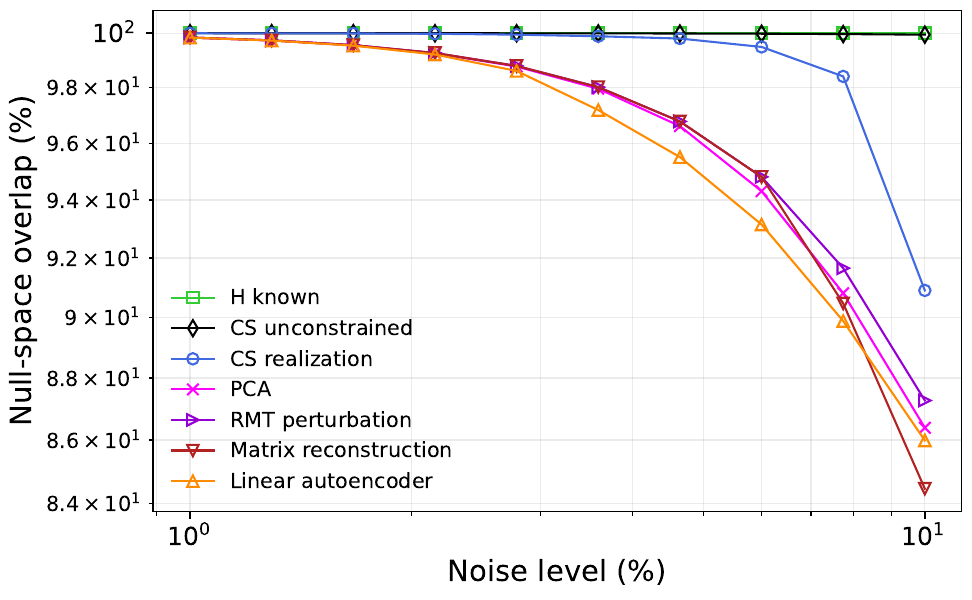}
    \caption{IEEE-118 null-space overlap versus branch-relative training
    noise.  Curves show the mean over 20 paired trials at each level.
    }
    \label{fig:noise_sensitivity_ieee118}
\end{figure}

In \Cref{fig:noise_sensitivity_ieee118}, \ac{cs} unconstrained remains essentially
at 100\%, since only its parameters are refitted.  \ac{cs} realization remains
above 99\% through approximately 6\% noise and falls to about 91\% at
10\% noise.  At 10\%, \ac{rmt} perturbation, \ac{pca}, the linear autoencoder, and
matrix reconstruction reach about 87.3\%, 86.4\%, 86.0\%, and 84.5\%,
respectively.  Hence the learned cycle structure achieves more robust
subspace recovery over the evaluated range, although the final drop reflects
discrete tree and support errors.

\Cref{fig:dc_state_impact,fig:noise_sensitivity_ieee118} measure different
properties.  \Cref{fig:dc_state_impact} fixes one model learned at 10\% noise
and measures \ac{bdd} acceptance against state impact in noise-normalized
coordinates.  \Cref{fig:noise_sensitivity_ieee118} re-estimates every method
over 20 training-noise realizations and measures unwhitened Euclidean
subspace geometry.  Their method ordering therefore need not coincide.

\begin{figure*}[!t]
    \centering
    \begin{minipage}[t]{0.49\textwidth}
        \centering
        \includegraphics[width=\linewidth]{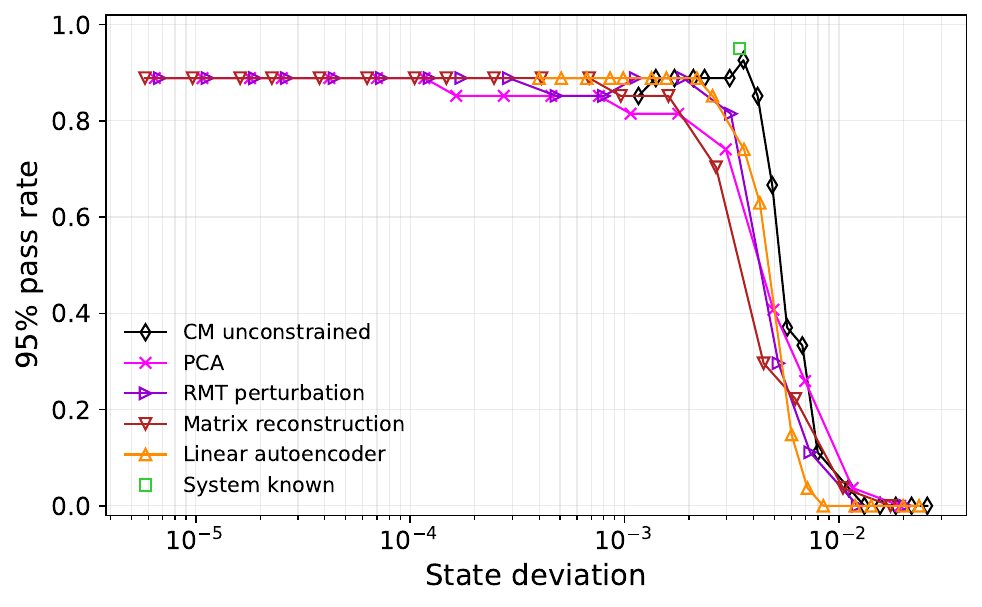}
        \vspace{-0.4em}
        {\footnotesize (a) IEEE 14-bus}
    \end{minipage}\hfill
    \begin{minipage}[t]{0.49\textwidth}
        \centering
        \includegraphics[width=\linewidth]{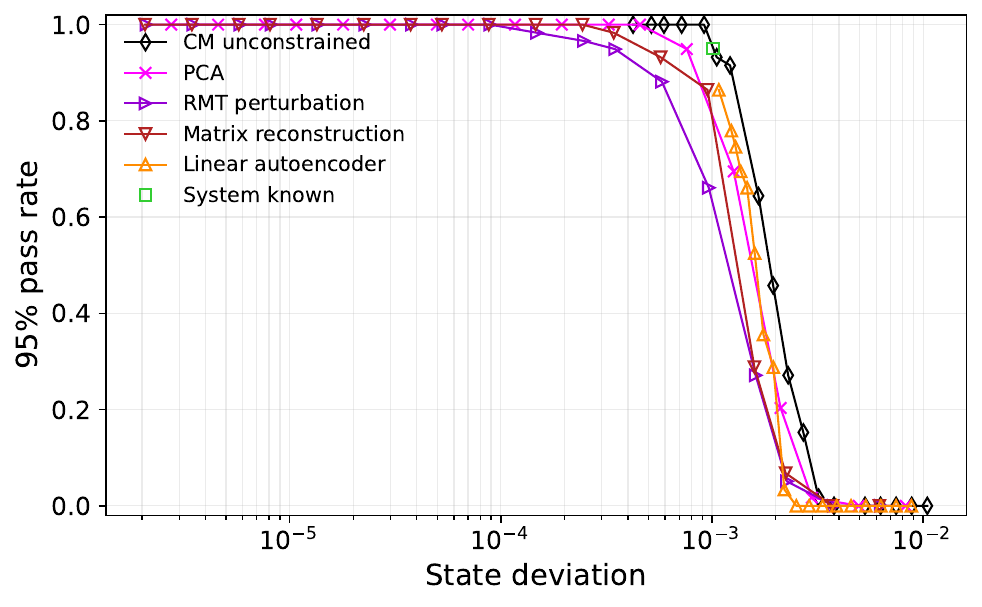}
        \vspace{-0.4em}
        {\footnotesize (b) IEEE 30-bus}
    \end{minipage}

    \vspace{0.3em}
    \begin{minipage}[t]{0.49\textwidth}
        \centering
        \includegraphics[width=\linewidth]{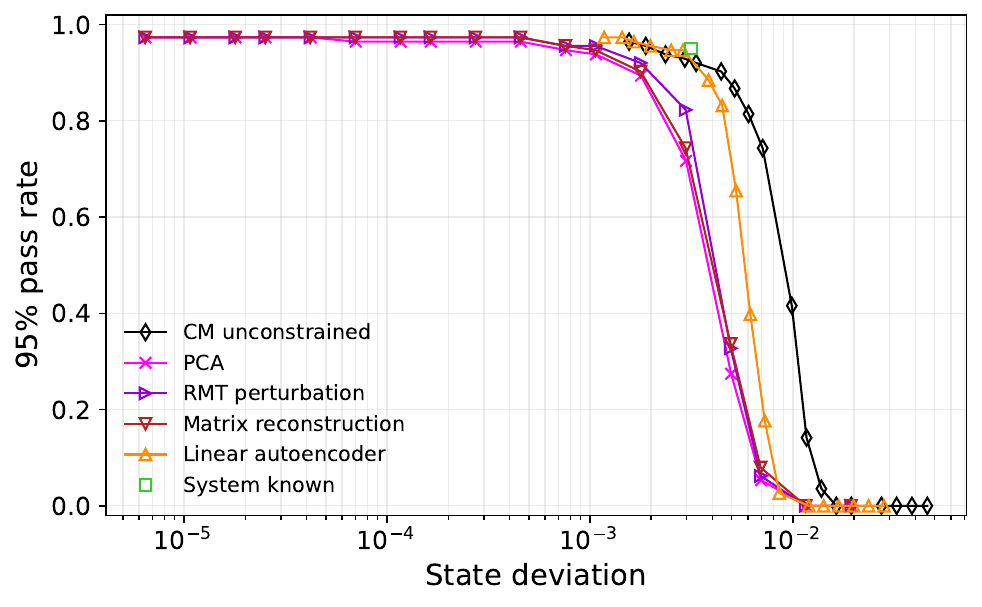}
        \vspace{-0.4em}
        {\footnotesize (c) IEEE 57-bus}
    \end{minipage}\hfill
    \begin{minipage}[t]{0.49\textwidth}
        \centering
        \includegraphics[width=\linewidth]{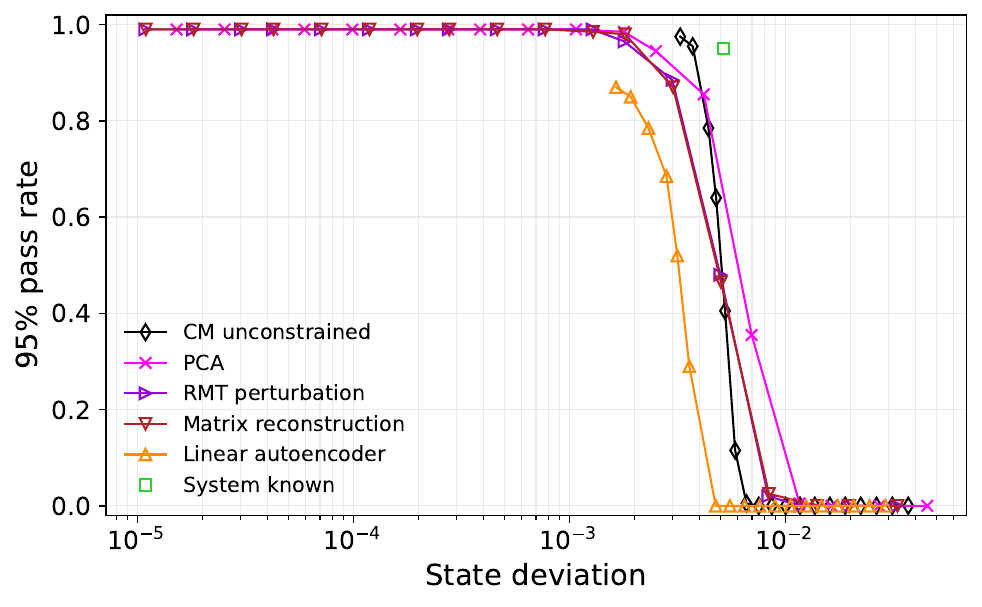}
        \vspace{-0.4em}
        {\footnotesize (d) IEEE 118-bus}
    \end{minipage}
    \caption{AC pass rate at a residual threshold calibrated for 95\%
    nominal acceptance versus mean state impact.  The system-known marker is a nominal noise-to-clean reference.}
    \label{fig:ac_state_impact}
\end{figure*}

\subsection{AC cycle-manifold results}

For AC, the threshold is the 95th percentile of the physical residual over
held-out noisy measurements.  If $\widehat V_i,\widehat\theta_i$ are the AC
state estimates for a generated measurement and
$\widehat V_i^{(0)},\widehat\theta_i^{(0)}$ are those for its noisy base, the
state impact is
\begin{equation}
\label{eq:ac_state_impact_metric}
\begin{aligned}
d_i^{\mathrm{AC}}
&=\frac{1}{\sqrt{2(\dimx-1)}}\left[
\sum_{k\ne\mathrm{ref}}(\Delta|V_{ik}|)^2\right.\\
&\left.\hspace{2.2em}+\sum_{k\ne\mathrm{ref}}
\operatorname{wrap}(\Delta\theta_{ik}-\Delta\theta_{i,\mathrm{ref}})^2
\right]^{1/2}.
\end{aligned}
\end{equation}
The angle difference removes the reference-angle gauge.  The plotted point
uses the mean of \cref{eq:ac_state_impact_metric} and the fraction whose
physical AC residual is below the fixed threshold, using the 
converged cohort across methods.

\Cref{fig:ac_state_impact} shows that the learned cycle manifold supports
nontrivial state changes while retaining a high pass rate and is competitive
with the tested subspace attacks on all four systems.  Its gain is clearest
on IEEE-30 and IEEE-57.  Because direct completion overwrites the noisy base
chords, \ac{cm} may have a small nonzero impact even as the tree perturbation tends
to zero.  These results validate the topology-assisted parameters construction;
they do not establish measurement-only AC manifold recovery.

To separate systematic completion error from trialwise variance, let
$e_{ic}=\widehat s_{ic}-s_{ic}^{\mathrm{clean}}$ on held-out chord $c$ and
define
\begin{equation}
    \label{eq:ac_chord_bias}
    b_{\mathrm{chord}}
    =\mathrm{baseMVA}\left(
      \frac{1}{q}\sum_{c=1}^{q}
      \left|\frac{1}{N}\sum_{i=1}^{N}e_{ic}\right|^2
      \right)^{1/2}.
\end{equation}

\begin{figure}[!t]
    \centering
    \includegraphics[width=\linewidth]{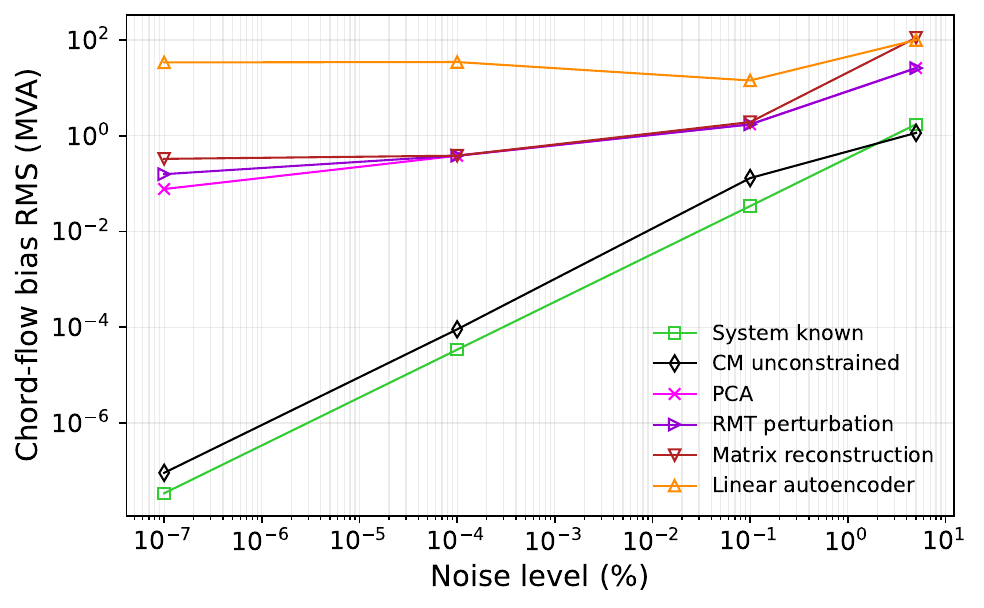}
    \caption{IEEE-118 chord-flow bias against measurement-noise level.  \ac{cm}
    unconstrained uses known topology but no physical admittances or voltage
    states.}
    \label{fig:ac_chord_bias}
\end{figure}

In \Cref{fig:ac_chord_bias}, \ac{cm} unconstrained follows the system-known
reference and its bias approaches numerical zero as measurement noise
vanishes, whereas the compared data-driven estimators retain nonzero bias.
Thus, under the correct topology, sufficient excitation, and a regular AC, \ac{cm} unconstrained has the zero-bias noiseless limit.  The nonlinear fit and the batch generation in these experiments are executed on a \ac{gpu}.

\section{Conclusion}
\label{sec:conclusion}
This paper proves that weighted cycle-space recovery is the necessary and sufficient information requirement for blind \ac{fdia} from branch-flow measurements, yielding $\H$ only up to 2-isomorphism and one  parameter scale per biconnected component.
The AC extension shows that  $P/Q$ measurements are parametersly described
by a necessary-and-sufficient cycle manifold and that this manifold can be
fitted from topology and measurements to generate parametersly consistent data.
Its lossless fixed-voltage small-angle linearization has the DC stealthy space
as its active-power tangent space and the weighted cycle space as its normal
space.
Learning the nonlinear fundamental-cycle structure without topology, and
therefore deriving a measurement-only AC blind \ac{fdia} upper bound, remains open.


\bibliographystyle{IEEEtran}
\bibliography{apssamp}

\end{document}